\documentclass[a4paper, 12pt, oneside]{article}
\usepackage{amsmath}
\usepackage{amssymb}
\usepackage{enumerate}
\usepackage{graphicx}
\usepackage{epsf}
\usepackage{verbatim}

\usepackage[amsmath,thmmarks]{ntheorem}
\newcounter{thm} 

\newtheorem{Thm}[thm]{Theorem}

\newtheorem{Lem}[thm]{Lemma}
\newtheorem{Prop}[thm]{Proposition}
\newtheorem{Cor}[thm]{Corollary}

\theoremstyle{nonumberplain}
\theoremheaderfont{\normalfont\itshape}
\theorembodyfont{\normalfont}
\theoremseparator{.}
\theoremsymbol{\ensuremath{\square}}
\newtheorem{proof}{Proof}

\def \Z {\mathbb Z}

\def \C {\mathbb C}

\def \d {\mathrm{d}}

\begin{document}
\title{Virasoro constraints and polynomial recursion for the linear Hodge integrals}
\author{Shuai Guo ,  Gehao Wang}
\date{}
\maketitle

\abstract
The Hodge tau-function is a generating function for the linear Hodge integrals. It is also a tau-function of the KP hierarchy. In this paper, we first present the Virasoro constraints for the Hodge tau-function in the explicit form of the Virasoro equations. The expression of our Virasoro constraints is simply a linear combination of the Virasoro operators, where the coefficients are restored from a power series for the Lambert W function. Then, using this result, we deduce a simple version of the Virasoro constraints for the linear Hodge partition function, where the coefficients are restored from the Gamma function. Finally, we establish the equivalence relation between the Virasoro constraints and polynomial recursion formula for the linear Hodge integrals.



\section{Introduction}
Let $\overline{M}_{g,n}$ be the moduli space of stable curves of genus $g$ with $n$ marked points, and $\psi_i$ be the first Chern class of the cotangent space over $\overline{M}_{g,n}$ at the $i$th marked point. Let $\lambda_j$ be the $j$th Chern class of the Hodge bundle over $\overline{M}_{g,n}$. The linear Hodge integrals are the intersection numbers of the form
$$\left<\lambda_j\tau_{d_1}\dots \tau_{d_n}\right>=\int_{\overline{M}_{g,n}}\lambda_j\psi_1^{d_1}\dots \psi_n^{d_n}.$$
They are defined to be zero when the numbers $j$ and $d_i$ do not satisfy the condition
\begin{equation}\label{condition}
j+\sum_{i=1}^n d_i=\dim(\overline{M}_{g,n})=3g-3+n.
\end{equation}
The Hodge tau-function $\exp(F_H(u,q))$ is a properly arranged generating function for the linear Hodge integrals. In \cite{MK}, Kazarian proved that it is a tau-function for the KP hierarchy. It is a basic fact that the space of all tau-functions of the KP hierarchy forms an orbit under the action of the so-called $\widehat{GL(\infty)}$ group. And this group is constructed via the exponential map from the infinite dimensional Lie algebra $\widehat{\mathfrak{gl}(\infty)}$ \cite{MJE}. The {\bf Virasoro operators} $L_m, m\in \Z,$ actually belong to this Lie algebra. They are in the form
$$L_m=\sum_{k>0,k+m>0} (k+m)q_k\frac{\partial}{\partial q_{k+m}}+\frac{1}{2}\sum_{a+b=m} ab\frac{\partial^2}{\partial q_a\partial q_b}+\frac{1}{2}\sum_{i,j>0,i+j=-m}q_iq_j,$$
and they form a representation of the Virasoro algebra. In fact, the Hodge tau-function $\exp(F_H(u,q))$ can be defined using the Virasoro operators. Let $\widetilde{\phi_0}(u,q)=q_1$, and for $k\geq 0$,
\begin{equation}\label{def:phi}
\widetilde{\phi_{k+1}}(u,q)=\left( L_{-2}+2uL_{-1}+u^2L_0-\frac{1}{2}q_1^2 \right)\cdot \widetilde{\phi_k}(u,q).
 \end{equation}
The function $F_H(u,q)$ is defined as
\begin{equation}\label{FHuq}
F_H(u,q)=\sum (-1)^j\left<\lambda_j\tau_{d_1}\dots \tau_{d_n}\right>u^{2j}\prod_{i=1}^n \widetilde{\phi_{d_i}}(u,q).
\end{equation}
Recently, Alexandrov introduced a method of constructing Virasoro constraints for the Hodge tau-function using the Kac-Schwarz operators \cite{AE}. However, the power series used to describe the constraints in \cite{AE} are rather complicated. In this paper, we first present a simpler version of the Virasoro constraints for the Hodge tau-function $\exp(F_H(u,q))$ using Virasoro operators. Later we will see in details that this simple version allows us to establish the equivalence relation between the Virasoro constraints and the polynomial recursion relation for the linear Hodge integrals.
\begin{Thm}\label{VforHodge}
Define the power series $\eta$ to be
\begin{equation}\label{eta}
\eta=\sqrt{2\log(1+z)-\frac{2z}{1+z}}.
\end{equation}
And for $m=-1,0,1,2,\dots$, we let
$$\frac{(1+z)^2}{z^2}\eta^{2m+2}=\sum_{i=2m}^{\infty}\nu_i^{(m)}z^i,  \quad\quad  -\frac{z}{1+z}\eta^{2m+2}=\sum_{j=2m+2}^{\infty}\gamma_j^{(m)}z^j.$$
\begin{enumerate}[(a)]
\item The Hodge tau-function $\exp(F_H(u,q))$ is annihilated by the infinite sequence of operators $\{V_m^{(H)}\}$ (Virasoro constraints) as
\begin{equation}\label{maineq}
V_m^{(H)}\cdot \exp(F_H(u,q))=0, \quad (m\geq -1)
\end{equation}
where
\begin{equation}
V_m^{(H)}=\sum_{i=2m}^{\infty} \nu_i^{(m)}u^{i-2m}L_i+\sum_{j=2m+3}^{\infty}\gamma_j^{(m)}ju^{j-2m-3}\frac{\partial}{\partial q_j}-\frac{u ^2}{24}\delta_{m,-1}+\frac{1}{8}\delta_{m,0},
\end{equation}
and $L_i$ is the Virasoro operator. The operators $\{V_m^{(H)}\}$ form a subalgebra of the Virasoro algebra, namely,
$$[V_m^{(H)},V_n^{(H)}]=2(m-n)V_{m+n}^{(H)}, \quad m,n\geq -1,$$
and they are generated by $V_{-1}^{(H)}$ and $V_2^{(H)}$.
\item The Hodge tau-function is annihilated by the operator $\mathfrak{D}^{(H)}$,
\begin{equation}\label{dilaton3}
\mathfrak{D}^{(H)}\cdot \exp(F_H(u,q))=0,
\end{equation}
where
\begin{equation}
\mathfrak{D}^{(H)}=L_0+u\frac{\partial}{\partial u}-3\sum_{k=0}^{\infty}(-1)^k\binom{k+2}{k}u^{k}\frac{\partial}{\partial q_{k+3}}+\frac{1}{8}.
\end{equation}
\end{enumerate}
\end{Thm}
We give two proofs of the above theorem. In Sect.\ref{Sec31} we give the first proof, using the Virasoro constraints for the Kontsevich-Witten tau-function and the main theorem in \cite{LW}. We introduce a method to compute the conjugation of Virasoro operators using power series in one variable $z$, and then prove Theorem \ref{VforHodge} without considering any knowledge of matrix model and Kac-Schwarz operators. Note that, after setting $u=0$, formula (\ref{maineq}) recovers the Virasoro constraints for the Kontsevich-Witten tau-function. The case $m=-1$ in formula (\ref{maineq}) will give us the {\bf string equation} for the linear Hodge integrals. And the case $m=0$ will lead to the statement (\emph{b}) in the above theorem, which implies the {\bf dilaton equation}. The expression of string and dilaton equations are well-known, and can be found in many papers such as \cite{GJV}. We will give a brief argument about how to obtain these two equations using Theorem \ref{VforHodge} in Appendix A.1.

Our first application of Theorem \ref{VforHodge} is to obtain the Virasoro constraints for the {\bf linear Hodge partition function} in our next result. The linear Hodge partition function $\exp(F_H(u,t))$ is a generating function for linear Hodge integral, where
\begin{equation}\label{Hodgeseries}
F_H(u,t)=\sum (-1)^j <\lambda_j \tau_0^{k_0}\tau_1^{k_1}\dots>u^{2j} \frac{t_0^{k_0}}{k_0!}\frac{t_1^{k_1}}{k_1!}\dots.
\end{equation}
The existence of Virasoro constraints for $\exp(F_H(u,t))$ has been confirmed in many papers, (see, for example, \cite{AMAM}, \cite{Z}), but these former results do not provide the Virasoro equations. Here, we present the Virasoro constraints for $\exp(F_H(u,t))$ in a simple equation form, where the coefficients involve a set of rational numbers $\{C_i\}$. These numbers appear as the coefficients of the Stirling's approximation of {\bf gamma function} $\Gamma(z)$:
\begin{equation}\label{gamma}
\Gamma(z)\sim z^{z-\frac{1}{2}}e^{-z}\sqrt{2\pi}\sum_{i=0}^{\infty}C_iz^{-i}.
\end{equation}
\begin{Cor}\label{VforPartition}
The linear Hodge partition function $\exp(F_H(u,t))$ is annihilated by the infinite sequence of operators $\{\widehat{V}_m^{(H)}\}$ (Virasoro constraints) as
\begin{equation*}
\widehat{V}_m^{(H)}\cdot \exp(F_H(u,t))=0, \quad (m\geq -1)
\end{equation*}
where
\begin{multline}\label{SimpleExpression}
\widehat{V}_m^{(H)}=\sum_{k= m}^{\infty}\sum_{i=0}^{k-m}c_i^{(k,m)}u^{2(m+i)}t_{k-m-i}\frac{\partial}{\partial t_k}\\
+\frac{1}{2}\sum_{n=0}^{\infty}(-u^2)^{m+n-1}\sum_{k=0}^{m+n-1}c_n^{(k,m)}\frac{\partial^2}{\partial t_k\partial t_{m+n-1-k}}\\
-\sum_{k=m+1}^{\infty}c_{k-m-1}^{(k,m)}u^{2(k-m-1)}\frac{\partial}{\partial t_k}-\frac{u^2}{24}\delta_{m,-1}+\frac{1}{8}\delta_{m,0},
\end{multline}
and
\begin{equation}\label{c}
c_n^{(k,m)}=\sum_{i=0}^n(-1)^i\frac{(2k-2i+1)!!}{(2k-2m-2i-1)!!}C_iC_{n-i}.
\end{equation}
Here we use the notation $(-2n-1)!!=\frac{(-1)^n}{(2n-1)!!}$.  The operators $\{\widehat{V}_m^{(H)}\}$ form a subalgebra of the Virasoro algebra, namely,
$$[\widehat{V}_m^{(H)},\widehat{V}_n^{(H)}]=2(m-n)\widehat{V}_{m+n}^{(H)}, \quad m,n\geq -1,$$
and they are generated by $\widehat{V}_{-1}^{(H)}$ and $\widehat{V}_2^{(H)}$.
\end{Cor}
In Sect.\ref{Sec32}, we prove this corollary.  And we will see that the strategy used in \cite{AMAM} and \cite{Z}, which is based on formula (\ref{1}) (cf. \cite{FP}, \cite{AG1}) derived from Mumford's theorem \cite{Mu}, will eventually result in our expression (\ref{SimpleExpression}). The numbers $\{C_i\}$ have some other interesting properties, which will be introduced in Sect.\ref{Sec2}.

It is commonly assumed that the generating function of enumerative problems enjoy similar properties which are related to each other, such as Virasoro constraints, integrable hierarchy, cut-and-join type equation and Eynard-Orantin recursion relation. For example, the Eynard-Orantin recursion for the $\psi$-class intersection numbers is exactly the DVV-formula \cite{DVV} for the intersection numbers, which is equivalent to the Virasoro constraints for the Kontsevich-Witten generating function. And other examples including simple Hurwitz numbers, Mirzakhani's recursion relation for the Weil-Petersson volumes of the moduli spaces of bordered Riemann surfaces \cite{MMBS} and Grothendieck's dessins d'enfants \cite{MZo} indicate the same phenomenon. 

In the case of linear Hodge integrals, the recursion relation has been established in \cite{MZ} by calculating the Laplace transform of the cut-and-join equation of the Hurwitz numbers \cite{MK}, (see Eq.(\ref{cutandjoin})). We present this recursion formula as Eq.(\ref{recursion}) in Sect.\ref{Sec4}. Furthermore, the direct image of the recursion formula by the projection map of the spectral curve on $\C$ is exactly the Bouchard-Mari\~no conjecture \cite{EMS}, which has been solved in \cite{BEMS} using matrix integral representation of simple Hurwitz numbers, and in \cite{EMS},\cite{MZ} using the recursion formula. For the Hodge tau-function, let us recall the following theorem from \cite{MK}, which is also derived from the cut-and-join equation of the simple Hurwitz numbers. 

\begin{Thm}[Kazarian, \cite{MK}] \label{Kazarian}
The tau-function $\exp(F_H(u,q))$ is subject to the differential equation
\begin{equation}\label{Kazariancutandjoin}
\frac{1}{3}u^{-2}\frac{\partial}{\partial u} \exp(F_H(u,q))=\widetilde{M}\cdot \exp(F_H(u,q)),
\end{equation}
where
\begin{multline}\label{Kazarianoperator}
\widetilde{M}=\vphantom{\frac{1}{2}}M_0+4u^{-1}M_{-1}+6u^{-2}M_{-2}+4u^{-3}M_{-3}+u^{-4}M_{-4}  \\
-\frac{4}{3}u^{-3}L_0-u^{-4}L_{-1}+\frac{1}{4}u^{-2}q_2+\frac{1}{3}u^{-3}q_3+\frac{1}{8}u^{-4}q_4.
\end{multline}
Here $M_k$ is the cut-and-join type operator. For $k\leq 0$, it is in the form
\begin{multline*}
M_k=\frac{1}{2}\sum_{\substack{i,j>0\\i+j+k>0}}(i+j+k)q_iq_j\frac{\partial}{\partial q_{i+j+k}}\\
+\frac{1}{2}\sum_{\substack{i,j>0\\i+j-k>0}}ijq_{i+j-k}\frac{\partial^2}{\partial q_i\partial q_j}
+\frac{1}{6}\sum_{\substack{i,j>0\\-k-i-j>0}}q_iq_jq_{-k-i-j},
\end{multline*}
and $L_i$ is the Virasoro operator.
\end{Thm}

Naturally, it is considered (mentioned in \cite{AE}) that the recursion relation for the linear Hodge integrals should be equivalent to its Virasoro constraints. In Sect.\ref{Sec4}, we will establish this equivalence relation by proving the following result.

\begin{Thm}\label{recursionequi} 
The following three different types of descriptions for linear Hodge integrals are equivalent:
\begin{enumerate}[(a)]
   \item The polynomial recursion formula for the linear Hodge integrals;
 \item Kazarian's formula \eqref{Kazariancutandjoin} in Theorem \ref{Kazarian} ;
    \item  Virasoro constraints \eqref{maineq} and Eq.\eqref{dilaton3} in Theorem \ref{VforHodge}.
\end{enumerate}
\end{Thm}
In our proof, we do not require any knowledge of Hurwitz number and its cut-and-join equation. In Sect.\ref{Sec41} we show that the recursion formula is equivalent to Kazarian's formula \eqref{Kazariancutandjoin}. Then, in Sect.\ref{Sec42}, we prove Kazarian's formula directly using Theorem \ref{VforHodge}. Thus, we have obtained an independent proof of the recursion formula. As a completion of establishing the equivalence relation, we give the second proof of Theorem \ref{VforHodge} in Sect.\ref{Sec43} using the following result, which allows us to produce the Virasoro constraints for the Hodge tau-function using Kazarian's formula. Note that we will use the coefficient operator $[z^i]$ throughout our context. That is,
$$F=\sum A_iz^i \quad\rightarrow \quad [z^i]F =A_i.$$
\begin{Prop}\label{epsiloncommutator}
For $m\geq 1$, let
\begin{equation}\label{epsilon}
E^{(m)} = \sum_{j=2m}^{\infty}\left[z^j\right](\eta^{2m})ju^{j-2m}\frac{\partial}{\partial q_j}.
\end{equation}
Then, we have
\begin{equation*}
V_{m-2}^{(H)}=-\frac{1}{2}[\widetilde{M}-\frac{1}{3}u^{-2}\frac{\partial}{\partial u}, u^{2m}E^{(m)} ].
\end{equation*}
\end{Prop}

\section{The tau-functions and their relations}\label{Sec2}
The intersections of the $\psi$-classes are evaluated by the integral:
$$\left<\tau_{d_1}\dots \tau_{d_n}\right>=\int_{\overline{M}_{g,n}}\psi_1^{d_1}\dots \psi_n^{d_n}.$$
The Kontsevich-Witten generating function in variables $t_k$ is defined as
$$F_K(t)=\sum <\tau_0^{k_0}\tau_1^{k_1}\dots>\frac{t_0^{k_0}}{k_0!}\frac{t_1^{k_1}}{k_1!}\dots.$$
From the linear Hodge partition function (\ref{Hodgeseries}), we can see that $F_H(0,t)=F_K(t)$. Now let
$$F_K(q)=F_K(t)|_{t_k=(2k-1)!!q_{2k+1}}.$$
It is well-known that $\exp{(F_K(q))}$ is a tau-function for the KdV hierarchy \cite{K}.

The function $F_H(u,q)$ defined in (\ref{FHuq}) can be obtained from $F_H(u,t)$ after the a change of variables $t_k=\widetilde{\phi_k}(u,q)$, that is,
$$F_H(u,q)=F_H(u,t)|_{t_k=\widetilde{\phi_k}(u,q)}.$$
The polynomials $\widetilde{\phi_k}(u,q)$ can also be described in the following way (cf. \cite{MK}). Let
\begin{equation}\label{D}
\widehat{D}=(u+z)^2z\frac{\partial}{\partial z}.
\end{equation}
Consider the following sequence of polynomials $\phi_k(u,z)$:
\begin{equation}\label{phiuz}
\phi_0(u,z)=z,\quad\quad \phi_{k}(u,z)=\widehat{D}^{k}\cdot z.
\end{equation}
Then $\widetilde{\phi_k}(u,q)$ can be obtained by replacing $z^m$ with $q_m$ in $\phi_k(u,z)$, and they satisfy Eq.(\ref{def:phi}). 
If we set $u=0$, the polynomial $\widetilde{\phi_k}(u,q)$ only has the term $(2k-1)!!q_{2k+1}$, and $\exp(F_H(0,q))$ is exactly the Kontsevich-Witten tau-function $\exp(F_K(q))$.

Next, we present some known relations between the tau-functions, which will be used to construct the Virasoro constraints in the next section. First we recall some important power series and results that have been introduced in the paper \cite{LW}. Let $h$ be the series defined as
\begin{equation}\label{h}
h=(\sum_{i=0}^{\infty} (-1)^{i}b_iz^i)^{-1} -1.
\end{equation}
where the numbers $\{b_i\}$ are uniquely determined by the following relation,
\begin{equation*}
(n+1)b_n=b_{n-1}-\sum_{k=2}^{n-1}kb_kb_{n+1-k}
\end{equation*}
with $b_1=1, b_2=1/3$. 

\vspace{6pt}
\noindent
{\bf Remark}: For the numbers $\{C_i\}$, we have, $C_i=(2i+1)!!b_{2i+1}$, \cite{MM}. From the stirling's approximation of gamma function, we have another presentation:
$$\log (\Gamma (z)) \sim \left(z-\tfrac{1}{2}\right)\log(z) -z + \tfrac{1}{2}\log(2 \pi) + \sum_{k=1}^\infty \frac{B_{2k}}{2k(2k-1)}z^{-2k+1},$$
which implies that
$$\exp\left(\sum_{k=1}^\infty \frac{B_{2k}}{2k(2k-1)}z^{-2k+1}\right) =\sum_{i=0}^{\infty}C_iz^{-i}.$$
Here $B_{2k}$'s are the Bernoulli numbers defined by:
$$\frac{t}{e^t-1}=\sum_{m=0}^{\infty} B_m\frac{t^m}{m!}.$$
Furthermore, for $i\geq 1$, $C_i$ can also be explicitly expressed by the following formula
$$C_i=\sum_{k=1}^{2i} (-1)^k \frac{d_3(2i+2k,k)}{2^{i+k}(i+k)!},$$
where $d_3(n,k)$ is defined to be the number of permutations on $n$ elements consisting of $k$ permutation cycles with length greater than or equal to $3$, (see \cite{LC}).

The series $h$ is a solution to the following equation
\begin{equation}\label{hLambert}
\frac{1}{1+h}e^{-\frac{1}{1+h}}=e^{-\frac{1}{2}z^2-1}.
\end{equation}
For $m>0$, let $\{a_m\}$ be the set of numbers determined by the following equation
\begin{equation}\label{ham}
\exp(\sum_{m>0}a_mz^{1+m}\frac{\d}{\d z})\cdot z=h.
\end{equation}
If we let the inverse function of $h$ be $\eta$, then
\begin{equation}\label{def:eta}
\eta=\exp(-\sum_{m>0}a_mz^{1+m}\frac{\d}{\d z})\cdot z,
\end{equation}
And, by Eq.(\ref{hLambert}), we can see that the expression of $\eta$ is exactly Eq.(\ref{eta}).

For $m>0$, the operators $L_m$ and the numbers $\{a_m\}$ appear in the following relation between Hodge tau-function and Kontsevich-Witten tau-function, \cite{LW}.
\begin{Thm}\label{KtoHodge}
For the two tau-functions $\exp(F_H(u,q))$ and $\exp(F_K(q))$, we have
\begin{equation}\label{KtoHodgeequation}
\exp(F_H(u,q))=\exp(U)\exp(P)\cdot\exp{(F_K(q))}.
\end{equation}
where
\begin{equation*}
U = \sum_{m>0} a_mu^mL_m, \quad P = -\sum_{k=1}^{\infty} b_{2k+1}u^{2k}\frac{\partial}{\partial q_{2k+3}}\nonumber.
\end{equation*}
\end{Thm}
Eq.(\ref{KtoHodgeequation}) is actually equivalent to Eq.(\ref{1}), which is deduced from Mumford's theorem (cf. \cite{Mu}, \cite{FP} and \cite{AG1}). Note that both of the operators $\exp(U)$ and $\exp(P)$ belong to $\widehat{GL(\infty)}$. Hence the KP integrability is preserved in this expression.

At the end of this section, we give an important property for a special type of generating function. Note that this is an independent result (not just for the Hodge tau-function):
\begin{Prop}\label{q2m}
Let $F(u,q)$ be a formal series in variables $u$ and $q_k$, $k\geq 1$. If $F(u,q)$ can be represented as a function of $u$ and $\{\widetilde{\phi_k}(u,q)\}$,
then for the operator $E^{(m)}$ defined in Eq. \eqref{epsilon}, $m\geq 1$,
\begin{equation}\label{epsilonConstraints}
E^{(m)}  \cdot F(u,q)=0.
\end{equation}
Especially, the Hodge tau function satisfies Eq. \eqref{epsilonConstraints} .
\end{Prop}
\begin{proof}
It suffices to prove
\begin{equation}\label{sufficetoprove}
E^{(m)}\cdot \widetilde{\phi_{k}}(u,q)=0,
\end{equation}
for $k\geq 0$. Let 
$$X_m=\sum_{k>0,k+m>0} (k+m)q_k\frac{\partial}{\partial q_{k+m}}, \quad X=\sum_{m\geq 1}a_mu^mX_m.$$ 
Then the operator $e^X$ acts as the change of variables (see Lemma 9 and 10 in \cite{LW}):
\begin{equation}\label{eqn:LW}
e^X\cdot q_{2n+1}=\frac{1}{(2n-1)!!}\sum_{i=0}^nC_iu^{2i}\widetilde{\phi_{n-i}}(u,q).
\end{equation}
Note that $\widetilde{\phi_{0}}(u,q)  = q_1$, and (see Eq.\eqref{qifkz}), $e^{ad_X} \left(2m\partial q_{2m}\right) = E^{(m)}.$
Then Eq.(\ref{sufficetoprove}) can be shown easily by doing induction on $k$ and using Eq.(\ref{eqn:LW}).
\end{proof} 

\section{Virasoro constraints}\label{Sec3}
The Virasoro constraints for the function $\exp(F_K(t))$ are  \cite{IZ}:
\begin{multline*}
\widehat{L}_m=\sum_{k\geq m}\frac{(2k+1)!!}{(2k-2m-1)!!}t_{k-m}\frac{\partial}{\partial t_k}+\frac{1}{2}\sum_{k+l=m-1} (2k+1)!!(2l+1)!!\frac{\partial^2}{\partial t_k\partial t_l}\\
-(2m+3)!!\frac{\partial}{\partial t_{m+1}}+\frac{t_0^2}{2}\delta_{m,-1}+\frac{1}{8}\delta_{m,0},  (m=-1,0,1,2,\dots)
\end{multline*}
such that, for $m\geq -1$,
$$\widehat{L}_m \cdot\exp(F_K(t)) =0.$$
Since the Kontsevich-Witten tau-function $\exp(F_K(q))$ has no even variables $q_{2k}$'s, we can deduce that
\begin{equation}\label{VforKW}
V^{(K)}_{2m}\cdot\exp(F_K(q))=0,
\end{equation}
where
\begin{equation*}
V^{(K)}_{2m}=L_{2m}-(2m+3)\frac{\partial}{\partial q_{2m+3}}+\frac{1}{8}\delta_{m,0}.
\end{equation*}

Generally speaking, for a tau-function $\tau$, if there exists an operator $A$, such that
$$\tau=e^{A}\cdot \exp(F_K(q)),$$
and the operator
\begin{equation}\label{conjugation}
e^{ad_{A}}V^{(K)}_{2m}=e^{A}V^{(K)}_{2m}e^{-A}
\end{equation}
is well-defined, then
$$(e^{ad_{A}}V^{(K)}_{2m})\cdot \tau=0,\quad m\geq -1.$$
And $\{e^{ad_{A}}V^{(K)}_{2m}\}$ form a set of Virasoro constraints for $\tau$. The formula
\begin{equation}\label{AM1}
e^{ad_{A}} M=\sum_{n=0}^{\infty} \frac{1}{n!} ad_{A}^nM\\
\end{equation}
can help us with the computation. Here `` $ad^j$ '' is the notation for the nested commutator:
$$ad_A^0M=M, \quad ad_A M=[A,M], \quad ad_A^j M=[A,ad_A^{j-1}M].$$
On the other hand, one must check that the conjugation (\ref{conjugation}) is well-defined, and does not give us a divergent result. This strategy guarantees the existence of the Virasoro constraints. However, one can see that the direct computation using (\ref{conjugation}) usually results in an infinite summation of differential operators with infinite many variables. Hence, seeking a better way to compute such conjugation becomes necessary. Based on the known relations introduced in the previous section, we will derive the Virasoro constraints $V^{(H)}_{m}$ for the Hodge tau-function in the next subsection using this strategy. And we will introduce a technique to compute such conjugation using power series in one variable.
 
\vspace{6pt}
\noindent
{\bf Remark:} In \cite{AE}, Alexandrov established the following formula  
$$\tau_{KW}=C(u)\widehat{G}_{+}\tau_{Hodge}.$$
that connects the Hodge and Kontsevich-Witten tau-functions, where $C(u)$ is an unknown Taylor series in $u$ and $\widehat{G}_{+}$ is a $\widehat{GL(\infty)}$ operator formed by Virasoro operators. The series $C(u)$ is conjectured to be one (Conjecture 2.1 in \cite{AE}). Hence this formula confirmed a weaker version of the conjecture raised in \cite{A}. Using this formula, Alexandrov obtained a set of Virasoro constraints for Hodge tau-function described by equations (2.193)-(2.196) in \cite{AE}. In \cite{LW}, Alexandrov's conjecture in \cite{A} has been proved using Theorem \ref{KtoHodge}. In fact, this theorem allows us to specify the prefactor $C(u)$, and hence prove Conjecture 2.1. But since we do not use the results and methodology in \cite{AE}, the proof of Conjecture 2.1 is not considered in this paper. These two formulas produce different expressions of the constraints (including the first order constraints $E^{(m)}$). For example, the operator $V_{-1}^{(H)}$ is in the form
\begin{equation*}
L_{-2}+2uL_{-1}+u^2L_0-\frac{u^2}{24}-\sum_{j=1}^{\infty}(-u)^{j-1}j\frac{\partial}{\partial q_j},
\end{equation*}
while the operator $\widehat{L}_{-1}^{Hodge}$ in (2.198) of \cite{AE} consists of infinite number of Virasoro operators. On the other hand, these two sets of constraints are equivalent to each other, since they are both obtained from the Virasoro constraints of Kontsevich-Witten tau function through conjugation.

\subsection{Proof of Theorem \ref{VforHodge}}\label{Sec31}
In this subsection, we prove Theorem \ref{VforHodge}. From the formula in Theorem \ref{KtoHodge}, we can see that the relation between two tau-functions has two parts, and we express the relation using the following diagram:
$$\exp(F_K(q)) \xrightarrow{\exp(P)} \exp(P)\cdot \exp(F_K(q)) \xrightarrow{\exp(U)} \exp(F_H(u,q)).$$
For the first part, since the operator $\exp(P)$ preserves the KP integrability, and does not include the variables $q_{2k}$, the function $\exp(P)\cdot \exp(F_K(q))$ is in fact a tau-function for the KdV hierarchy. Now, we consider the conjugation $e^PV^{(K)}_{2m} e^{-P}.$ The operator $\exp(P)$ acts as a shift on variables $q_{2k+3}$, $(k\geq 1)$. Hence
$$e^P V^{(K)}_{2m} e^{-P}=L_{2m}+P_m^{(1)}+\frac{1}{8}\delta_{m,0},$$
where
\begin{equation*}
P_m^{(1)}=-\sum_{k=0}^{\infty}(2m+2k+3)b_{2k+1}u^{2k}\frac{\partial}{\partial q_{2m+2k+3}}.
\end{equation*}
Let
\begin{equation}\label{Pm2}
P_m^{(2)}=\sum_{i=1}^{\infty}(2m+2+i)(-1)^ib_{i}u^{i-1}\frac{\partial}{\partial q_{2m+2+i}}.
\end{equation}
Since $e^P\cdot e^{F_K(q)}$ has no even variables $q_{2k}$, we can deduce that
\begin{equation}\label{L2m}
(L_{2m}+P_m^{(2)}+\frac{1}{8}\delta_{m,0}) \cdot \left(e^P\cdot e^{F_K(q)}\right)=0.
\end{equation}
The operators $\{L_{2m}+P_m^{(2)}+\frac{1}{8}\delta_{m,0}\}$ form a set of Virasoro constraints for the tau-function $\exp(P)\cdot \exp(F_K(q))$. Next, we set
\begin{equation}\label{VmHoperator}
V^{(H)}_m=e^{U} (L_{2m}+P_m^{(2)}+\frac{1}{8}\delta_{m,0}) e^{-U}.
\end{equation}
Then we can conclude that
\begin{Lem}\label{lem:VmHoperator}
If the right hand side of Eq.(\ref{VmHoperator}) is well-defined, then $\{V^{(H)}_m\}$ is a set of Virasoro constraints for $\exp(F_H(u,q))$.
\end{Lem}
In the rest of this subsection, not only will we see that Eq.(\ref{VmHoperator}) is well-defined, but also we will compute the explicit expression of $V^{(H)}_m$.

Let $V^{(H)}_m=L_m^{(H)}+P_m^{(H)}+\frac{1}{8}\delta_{m,0}$, where
$$P_m^{(H)}=e^{U}P_m^{(2)}e^{-U} \quad\mbox{and}\quad L_m^{(H)}=e^{U}L_{2m}e^{-U}.$$
(The reason why we choose $P_m^{(2)}$ instead of $P_m^{(1)}$ in Eq.(\ref{VmHoperator}) is to obtain a simplified version of $P^{(H)}_m$ as in Theorem \ref{VforHodge}). First we compute $P_m^{(H)}$. Let $U=X+Y$, where $X$ is defined in the proof of Proposition \ref{q2m}, and $Y$ is the second order derivative part.
Since $Y$ commute with the first order differential operator $\partial q_k$, and ${\rm ad}_{X}^n\partial q_k$ is always a first order differential operator for $n\geq 1$, we have
\begin{equation*}
P_m^{(H)}=e^{X}P_m^{(2)}e^{-X}=\sum_{n=0}^{\infty}\frac{1}{n!}{\rm ad}_{X}^n P_m^{(2)}.
\end{equation*}
We can see that ${\rm ad}_{X}^n P_m^{(2)}$ consists of operators $\partial_{q_k}$ with $k\geq n+1$. Hence the operator $P_m^{(H)}$ obtained from the above equation is well defined. Next, for $m\geq -1$, we consider $L_m^{(H)}$. Recall that $U=\sum_{m=1}^{\infty} a_mu^mL_m,$
and the Virasoro operators satisfy the commutator relation
$$[L_m,L_n]=(m-n)L_{m+n}+\frac{1}{12}(m^3-m)\delta_{m+n,0},$$
for $m,n\in {\mathbb Z}$. Using the formula (\ref{AM1}), we can see that this conjugation will give us a well-defined linear combination of $\{L_i\}$ with $i\geq 2m$, plus a constant term $-u^2/24$ when $m=-1$. In order to obtain the explicit expression of the operators $P_m^{(H)}$ and $L_m^{(H)}$, we introduce the next lemma, which tells us that the coefficients of the operator can be restored from a power series. 
\begin{Lem}\label{pmH+LmH}
\begin{enumerate}[(a)]
\item For $m\geq -1$, let
\begin{equation}\label{definepmz}
p_m(z)=\sum_{i=1}^{\infty} (-1)^ib_iz\frac{\d}{\d z} \left(\exp(-\sum_{k=1}^{\infty}a_kz^{1+k}\frac{\d}{\d z}) \cdot z^{2m+2+i}\right).
\end{equation}
Then,
$$\left[u^i\frac{\partial}{\partial q_i}\right]\left(u^{2m+3}P_m^{(H)}\right)=\left[z^i\right]p_m(z), \quad i\geq 2m+2.$$
\item For $m\geq -1$, let
\begin{equation}\label{defineGmz}
G_m(z)=\exp(-\sum a_kz^{1+k}\frac{\d}{\d z}+\sum ka_kz^{k})\cdot z^{2m}.
\end{equation}
Then,
$$\left[u^iL_i\right]\left(u^{2m}L_m^{(H)}\right)=\left[z^i\right] G_m(z), \mbox{ for } i\geq 2m.$$
\end{enumerate}
\end{Lem}
\begin{proof}
\begin{enumerate}[(a)]
\item We consider the set of power series $f_k(z)$ defined as
\begin{equation}\label{fkz}
f_k(z)=z\frac{\d}{\d z} \left(\exp(-\sum_{m=1}^{\infty}a_mz^{1+m}\frac{\d}{\d z}) \cdot z^{k}\right),
\end{equation}
for $k\geq 1.$ 
If we expand the right hand side of the above equation, we obtain
\begin{align*}
f_k(z)=& kz^k+\sum_{n=1}^{\infty}\frac{(-1)^n}{n!}\sum_{m_j>0}ka_{m_1}(k+m_1)a_{m_2}(k+m_1+m_2)\dots\\
&\quad\quad \dots a_{m_n} (k+\sum_{j=1}^n m_j) z^{k+\sum_{j=1}^n m_j}.
\end{align*}
A straightforward computation shows that
\begin{align*}
e^{X}\left(ku^k\frac{\partial}{\partial q_{k}}\right)e^{-X}
=&ku^k\frac{\partial}{\partial q_{k}}+\sum_{n=1}^{\infty}\frac{(-1)^n}{n!}\sum_{m_1,\dots,m_n>0} ka_{m_1}(k+m_1)\dots\\
&\quad\quad  \dots a_{m_n} (k+\sum_{j=1}^n m_j)u^{k+\sum_{j=1}^n m_j}\frac{\partial}{\partial q_{k+\sum_{j=1}^n m_j}}
\end{align*}
The above two equations give us
\begin{equation}\label{qifkz}
\left[u^i\frac{\partial}{\partial q_i}\right] \left(e^{X}\left(ku^k\frac{\partial}{\partial q_{k}}\right)e^{-X}\right)=\left[z^i\right]f_k(z), \mbox{ for } i\geq k.
\end{equation}
Then (a) follows from the definitions of $P_m^{(2)}$ by Eq.(\ref{Pm2}) and $p_m(z)$ by Eq.(\ref{definepmz}).

\item Similarly, we can expand the right hand side of Eq.(\ref{defineGmz}) and compare it with the general expression of $u^{2m}L_m^{(H)}$. 

\end{enumerate}
\end{proof}

Next, we compute $p_m(z)$ and $G_m(z)$. Using the Zassenhaus formula (see \cite{CMN}, \cite{LW}), we have
$$\exp(-\sum a_kz^{1+k}\frac{\d}{\d z}+\sum ka_kz^{k})=\exp(-\sum a_kz^{1+k}\frac{\d}{\d z})\exp(g(z)),$$
where $g(z)$ is a power series with no constant term. The differential operator on the right hand side of the above equation performs a change of variables $z\rightarrow \eta(z)$. Then
\begin{Prop}
\begin{enumerate}[(a)]\label{4}
  \item \begin{equation}\label{pmz}
  p_m(z)= z\frac{\d}{\d z} \left(-\frac{z}{1+z}\eta^{2m+2} \right).
  \end{equation}
  \item $$g(z)=2\log\frac{z(1+h)}{h},\quad
    G_m(z)=\frac{(1+z)^2}{z^2}\eta^{2m+2}.$$
\end{enumerate}
\end{Prop}
\begin{proof}
\begin{enumerate}[(a)]
\item By Eq.(\ref{def:eta}), we have
\begin{align*}
p_m(z)=& z\frac{\d}{\d z}\left(\eta^{2m+2} \sum_{i=1}^{\infty} (-1)^ib_i\eta^{i}\right)\\
\end{align*}
Since $\eta$ is the inverse function of $h$, using Eq.(\ref{h}), we can obtain Eq. (\ref{pmz}).
\item Let
$$\eta(u,z)=\exp(-\sum a_kz^{1+k}u^k\frac{\partial}{\partial z})\cdot z.$$
Then
\begin{align*}
\frac{\partial}{\partial u}\eta(u,z)
=&\left(\sum_{n=1}^{\infty} (-1)^{n}\frac{4}{(n+2)(n+1)n}u^{n-1}z^{1+n}\frac{\partial}{\partial z}\right)\cdot\eta(u,z).
\end{align*}
Hence, (see the argument of Eq.(\ref{gz}) in Appendix A.2),
$$z\frac{\d}{\d z}g(z)=\exp(\sum a_kz^{1+k}\frac{\d}{\d z})\cdot \left(-\sum_{n=1}^{\infty}d_{n}nz^{n}\right),$$
where
$$d_{n}=(-1)^{n}\frac{4}{(n+2)(n+1)n}, \quad (n\geq 1).$$
By Eq.(\ref{ham}), we have
\begin{align*}
\exp(\sum a_kz^{1+k}\frac{\d}{\d z})\cdot\left(-\sum_{n=1}^{\infty}d_{n}nz^{n}\right)
=&-2(\frac{1}{h}z\frac{dh}{dz}-\frac{1}{1+h}z\frac{dh}{dz}-1).
\end{align*}
And we can conclude that, $g(z)=2\log\frac{z(1+h)}{h}$, and
\begin{align*}
G_m(z)=\exp(-\sum a_kz^{1+k}\frac{\d}{\d z})\cdot \left(e^{g(z)} z^{2m}\right)=\frac{(1+z)^2}{z^2}\eta^{2m+2}.
\end{align*}
\end{enumerate}
\end{proof}
Now, if we let
\begin{equation*}
\gamma_j^{(m)}=\left[z^j\right]\frac{-z}{1+z}\eta^{2m+2}; \quad
\nu_i^{(m)}=\left[z^i\right]G_m(z),
\end{equation*}
then, by Lemma \ref{pmH+LmH} and the above proposition, we have
\begin{equation*}
P_m^{(H)}=\sum_{j=2m+3}^{\infty}\gamma_j^{(m)}ju^{j-2m-3}\frac{\partial}{\partial q_j},\quad L_m^{(H)}=\sum_{i=2m}^{\infty} \nu_i^{(m)}u^{i-2m}L_i-\frac{u^2}{24}\delta_{m,-1}.
\end{equation*}
From Lemma \ref{lem:VmHoperator}, we can see that
$$\left(L_m^{(H)}+P_m^{(H)}+\frac{1}{8}\delta_{m,0}\right)\cdot \exp(F_H(u,q))=0.$$
Therefore, we have proved the statement ({\emph a}) in Theorem \ref{VforHodge}.

Next we prove the statement ({\emph b}) of Theorem \ref{VforHodge}. Since $F_K(q)$ does not depend on parameter $u$,  we have $u\partial_u\exp(F_K(q))=0$. Then, after computing $\exp(P)(u\partial_u)\exp(-P)$, we get
$$\left(u\frac{\partial}{\partial u}-\sum_{i=1}^{\infty}(-1)^ib_iu^{i-1}(i-1)\frac{\partial}{\partial q_{i+2}} \right)\cdot \left(e^P\cdot e^{F_K(q)}\right)=0.$$
By Eq.(\ref{L2m}), we can deduce that
\begin{equation}\label{L0u}
\left(L_0+u\frac{\partial}{\partial u}+3\sum_{i=1}^{\infty}(-1)^ib_iu^{i-1}\frac{\partial}{\partial q_{i+2}}+\frac{1}{8} \right)\cdot \left(e^P\cdot e^{F_K(q)}\right)=0.
\end{equation}
Observe that
$$e^{ad_U}(L_0+u\frac{\partial}{\partial u})=L_0+u\frac{\partial}{\partial u};$$
$$\left[u^i\frac{\partial}{\partial q_i}\right]e^{ad_U} \left( u^k\frac{\partial}{\partial q_k}\right) =\left[z^i\right]\left\{\exp\left(-\sum_{m=1}^{\infty}a_m(z^{1+m}\frac{\d}{\d z}+mz^m)\right)\cdot z^k\right\}.$$
And we have the series
\begin{equation*}
\exp\left(-\sum_{m=1}^{\infty}a_m(z^{1+m}\frac{\d}{\d z}+mz^m)\right)\cdot \sum_{i=1}^{\infty}(-1)^ib_iz^{i+2}
=-\sum_{k=0}^{\infty}(-1)^k\binom{k+2}{k}z^{k+3}.
\end{equation*}
This proves Eq.(\ref{dilaton3}). Therefore, we have completed the proof of Theorem \ref{VforHodge} in this subsection.

\subsection{Proof of Corollary \ref{VforPartition}}\label{Sec32}
Before proving Corollary \ref{VforPartition}, we need to set up some tools. And the tools are once again the power series. Recall the operator $\widehat{D}$ and the definition (\ref{phiuz}) in Sect.\ref{Sec2}. Let
\begin{equation*}
D = (1+z)^2z\frac{\d}{\d z},\quad
\phi_{n}(z) = D^n\cdot z=\phi_n(u=1,z).
\end{equation*}
Let $f=\left(\eta(z^{-1})\right)^{-1}$. The series $f$ satisfy the following equation
\begin{equation*}
e^{-\frac{1}{2}f^{-2}}=\frac{z}{1+z}e^{\frac{1}{1+z}}.
\end{equation*}
In fact, many properties of the series $f$ can be found in the paper \cite{LW}. Here, we state some of the results that will used in our later proofs. First, we have $D\cdot f=f^3$. If we let $y=\frac{1}{2}f^{-2}$, $v=\frac{z}{1+z}$, and denote also by $D$ the vector field $D=(1+z)^2z\partial_z$, then
\begin{equation}\label{v}
v=1+\sum_{i=1}^{\infty}b_i(2y)^{i/2},
\end{equation}
and
\begin{equation}\label{vectorfield}
D=-\partial_y， \quad z=-\partial_y v.
\end{equation}
Let $F$ be a power series $F=z^n+\sum_{i=1}^{\infty} A_iz^{n-i},$
for some coefficients $A_i$. We define $(F)_{+}$ to be $\left(F\right)_{+}=z^n+\sum_{i=1}^{n-1} A_iz^{n-i}.$
Then
\begin{equation}\label{f}
\left(f^{2n+1}\right)_{+}=\frac{1}{(2n-1)!!}\sum_{i=0}^n C_i\phi_{n-i}(z).
\end{equation}
And, conversely, the polynomials $\phi_{n}(z)$ can be represented by the series $f$ as (see Lemma \ref{Lemphinz} in Appendix A.2):
\begin{equation}\label{phinz}
\phi_{n}(z) =\sum_{i=0}^n(-1)^i(2n-2i-1)!!C_i\left(f^{2n-2i+1}\right)_{+}.
\end{equation}

For the rest of this subsection, we consider the action of the differential operators $L_m^{(H)}$ and $P_m^{(H)}$ on $t_k$, where $t_k=\widetilde{\phi_k}(u,q)$. We write $L_m^{(H)}+\frac{u^2}{24}\delta_{m,-1}$ as two parts:
$$L_m^{(H)}+\frac{u^2}{24}\delta_{m,-1}=\sum_{i=2m}^{\infty}\nu_i^{(m)}u^iX_i+\frac{1}{2}\sum_{i=2m}^{\infty}\nu_i^{(m)}u^iY_i,$$
where
$$X_i=\sum_{k>0}(k+i)q_k\frac{\partial}{\partial q_{k+i}}\quad \mbox{and} \quad Y_i=\sum_{a+b=i}ab\frac{\partial^2}{\partial q_a\partial q_b}.$$
Now, to prove Corollary \ref{VforPartition}, we mainly need to prove the following lemma.
\begin{Lem}
For the number $c_n^{(k,m)}$ defined in Eq.(\ref{c}), we have
\begin{enumerate}[(a)]
\item
\[
[\sum_{i=2m}^{\infty}\nu_i^{(m)}u^iX_i,t_k]
=
\begin{cases}
0 & k<m \\
\sum_{n=0}^{k-m}c_n^{(k,m)}u^{2(m+n)}t_{k-m-n} & k\geq m .
\end{cases}.
\]
\item
\begin{equation}\label{commutator2}
[\sum_{i=2m}^{\infty}\nu_i^{(m)}u^iY_i,t_k]
=\sum_{\substack{l\geq 0\\l\geq m-k-1}}(-1)^{l+1}c_{k-m+l+1}^{(k,m)}u^{2k+2l+2}\frac{\partial}{\partial t_l}.  \nonumber
\end{equation}
\item
\[
[P_m^{(H)}, t_k]=
\begin{cases}
0 & k\leq m \\
-u^{2k-2m-2}c_{k-m-1}^{(k,m)}  & k> m .
\end{cases}.
\]
\end{enumerate}
\end{Lem}
\begin{proof}
\begin{enumerate}[(a)]
\item
Observe that, for $k>i$, $[X_i, q_k]=X_i\cdot q_k=kq_{k-i}$ and $z^{1-i}\frac{\d}{\d z} z^k=kz^{k-i}$. Then it is sufficient to consider the following equation
\begin{align*}
\left(\sum_{i=2m}^{\infty}\nu_i^{(m)}z^{1-i}\frac{\d}{\d z} \phi_k(z)\right)_{+}=\left(\frac{(1+z^{-1})^2}{z^{-2}}f^{-2m-2} z\frac{\d}{\d z} \phi_k(z)\right)_{+}.
\end{align*}
Here $\frac{(1+z^{-1})^2}{z^{-2}}f^{-2m-2}$ is a power series in $z^{-1}$ and $z\frac{\d}{\d z} \phi_k(z)$ is a finite degree polynomial in $z$. For $k\geq m$, we have
\begin{align*}
&\left(\frac{(1+z^{-1})^2}{z^{-2}}f^{-2m-2} z\frac{\d}{\d z} \phi_k(z)\right)_{+}\\
=&\sum_{i=0}^k(-1)^i(2k-2i+1)!!C_i\left(f^{2k-2m-2i+1}\right)_{+}\\
=&\sum_{i=0}^{k-m}(-1)^i\frac{(2k-2i+1)!!}{(2k-2m-2i-1)!!}C_i\sum_{j=0}^{k-m-i}C_j\phi_{k-m-i-j}(z),
\end{align*}
where the two steps follow from Eq.(\ref{phinz}) and Eq.(\ref{f}) respectively. The degree of parameter $u$ can be easily checked by comparing both sides of the equation in the lemma.
\item From the change of variables $t_k=\widetilde{\phi_k}(u,q)$, we can deduce that
$$b\frac{\partial}{\partial q_b}=\sum_{l\geq \frac{b-1}{2}}[z^{b-1}]\left(\frac{\d}{\d z}\phi_l(z)\right)u^{2l+1-b}\frac{\partial}{\partial t_l}.$$
On the other hand, for $a,b,i>0$ and $a+b=i$,
\begin{align*}
[\sum_{i=2m,i>0}^{\infty}\nu_i^{(m)}u^iY_i,u^{2k+1-a}q_a]&=\sum_{i=2m,i>0}^{\infty}\nu_i^{(m)}u^{2k+1+b}a\left(b\frac{\partial}{\partial q_b}\right).
\end{align*}
The degree of the parameter $u$ for the term $\frac{\partial}{\partial t_l}$ is $2k+2l+2$. Now, observe that,
\begin{equation*}
\frac{(1+z^{-1})^2}{z^{-2}}f^{-2m-2}(z\frac{\d}{\d z}z^a)=\sum_{i=2m}^{\infty}\nu_i^{(m)}az^{-i+a}.
\end{equation*}
Then, the coefficient of the term $z^{-1}$ in the following Laurent series
$$\frac{(1+z^{-1})^2}{z^{-2}}f^{-2m-2}\left(z\frac{\d}{\d z}\phi_k(z)\right)\left(\frac{\d}{\d z}\phi_l(z)\right)$$
is the constant part of the coefficients (not including $u$) of $\frac{\partial}{\partial t_l}$ in the commutator. Now we compute the residue of the above series at $z=0$. Using the relation (\ref{vectorfield}), we have
\begin{align*}
&\operatorname{Res}_{z=0}\left\{\frac{(1+z^{-1})^2}{z^{-2}}f^{-2m-2}\left(z\frac{\d}{\d z}\phi_k(z)\right)\left(\frac{\d}{\d z}\phi_l(z)\right)\right\}\\
=&\operatorname{Res}_{z=0}\left\{f^{-2m-3}\left(D^{k+1}\cdot z\right)\left(D^{l+1}\cdot z\right)\left(\frac{\d}{\d z}f\right)\right\}\\
=&2(-1)^{k+l}\operatorname{Res}_{y=0}(2y)^{m+1}\left(\partial_y^{k+1}z\right)\left(\partial_y^{l+1}z\right)\left(\frac{\d}{\d z}y\right)\\
=&(-1)^{k+l}2^{m+2}[y^{-m-2}]\left(\partial_y^{k+2}v\right)\left(\partial_y^{l+2}v\right).
\end{align*}
By Eq.(\ref{v}), and the fact that $C_i=(2i+1)!!b_{2i+1}$, we have
\begin{align*}
\partial_y^{k+2}v
&=\sum_{i=0}^{\infty}2^{\frac{2i+1}{2}-k-2}(-1)^{k-i-1}(2k-2i+1)!!C_iy^{\frac{2i+1}{2}-k-2}\\
&\quad\quad+ \sum_{j=k+2}^{\infty} 2^jj(j-1)\dots (j-k-1)b_{2j}y^{j-k-2}.
\end{align*}
Then,
\begin{align*}
&(-1)^{k+l}2^{m+2}[y^{-m-2}]\left(\partial_y^{k+2}v\right)\left(\partial_y^{l+2}v\right)\\
=&\sum_{i=0}^{k-m+l+1} (-1)^{l+1-i}\frac{(2k-2i+1)!!}{(2k-2m-2i-1)!!}C_iC_{k-m+l+1-i}.
\end{align*}
\item The case $k\leq m$ is obvious. Recall the series $p_m(z)$ defined as Eq.(\ref{pmz}). We can easily deduce that, for $k>m$,
\begin{align*}
P_m^{(H)} \cdot \widetilde{\phi_k}(u,q)
&=u^{2k-2m-2}[z^0]\left( -\frac{z^{-3}}{(1+z^{-1})^3}f^{-2m-2} \phi_{k+1}(z)\right).
\end{align*}
From what we have discussed in the proof of (\emph{a}), we have
$$\left(f^{-2m-2} \phi_{k+1}(z)\right)_{+}=\sum_{n=0}^{k-m}c_n^{(k,m)}\phi_{k-m-n}(z).$$
Since, for $n\geq 2$, $\phi_n(z)$ can be expressed as
$$\phi_n(z)=z(1+z)^{n+1}\kappa_n(z),$$
where $\kappa_n(z)$ is a polynomial containing the constant term $1$, we have
\begin{align}\label{cn}
&[z^{0}] \left\{\frac{z^{-3}}{(1+z^{-1})^3} \left(z(1+z)^{n+1}\kappa_n(z)\right)\right\}=0.
\end{align}
Also,
$$[z^0]\left(-\frac{z^{-3}}{(1+z^{-1})^3}\phi_0(z)\right)=0.$$
Therefore,
\begin{equation*}
[z^0]\left( -\frac{z^{-3}}{(1+z^{-1})^3}f^{-2m-2} \phi_{k+1}(z)\right)= -c_{k-m-1}^{(k,m)}.
\end{equation*}
\end{enumerate}
\end{proof}
Combining the results of the above lemma will give us the expression (\ref{SimpleExpression}) of operator $\widehat{V}_m^{(H)}$. The commutator relation of $\{\widehat{V}_m^{(H)}\}$ follows directly from $\{V_m^{(H)}\}$. Hence we have proved Corollary \ref{VforPartition} in this subsection.

\vspace{6pt}
\noindent
{\bf Remark:} For the two functions $F_K(t)$ and $F_H(u,t)$, they are related by operator $\exp(W)$ as (cf. \cite{FP}, \cite{AG1}):
\begin{equation}\label{1}
\exp(F_H(u,t))=e^W\cdot \exp(F_K(t)),
\end{equation}
where
\begin{equation}\label{W}
W=-\sum_{k\geq 1} \frac{B_{2k}u^{2(2k-1)}}{2k(2k-1)}(\frac{\partial}{\partial t_{2k}}-\sum_{i\geq 0} t_i\frac{\partial}{\partial t_{i+2k-1}}+\frac{1}{2}\sum_{i+j=2k-2} (-1)^i\frac{\partial^2}{\partial t_i\partial t_j}).
\end{equation}
Note that the operator $W$ does not belong to the $\widehat{\mathfrak{gl}(\infty)}$ algebra.

Recall the Virasoro constraints $\widehat{L}_m$ for $\exp(F_K(t))$. Let
$$\widehat{V}_m^{(H)}=e^W\widehat{L}_me^{-W}.$$
Then $\{\widehat{V}_m^{(H)}\}$ form a set of Virasoro constraints for the linear Hodge partition function $\exp(F_H(u,t))$. However, the operator $W$ and $\widehat{L}_m$ does not seem to have a simple commutation relation like the operators $L_i$. Hence a direct computation of the conjugation might be very complicated, and it is not obvious that we can find a power series to represent such conjugation like what we do in the previous subsection. Since Eq.(\ref{KtoHodgeequation}) is transformed from Eq.(\ref{1}) under the variable change $t_k=\widetilde{\phi_k}(u,q)$ (\cite{LW}), if we change the variables back to $t_k$ in the Virasoro equations (\ref{maineq}), we will eventually obtain the Virasoro constraints $\widehat{V}_m^{(H)}$ as defined before. Therefore, our approach above is to consider the action of the operators $V_m^{(H)}$ after we switch to the variables $t_k$, and write down the explicit expression of operators $\widehat{V}_m^{(H)}$. Although this method looks like a detour, it allows us to obtain the simple version (\ref{SimpleExpression}) of $\widehat{V}_m^{(H)}$, and hence prove Corollary \ref{VforPartition}.

\section{Proof of Theorem \ref{recursionequi}}\label{Sec4}
In this section we first give a brief introduction to the polynomial recursion formula for linear Hodge integral. The equivalence relation between the Virasoro constraints and the recursion formula will be bridged using Theorem \ref{Kazarian}. In specific, we show the equivalence relation between Theorem \ref{Kazarian} and Theorem \ref{Thmrecursion} in Sect.\ref{Sec41}, and then connect the Virasoro constraints and Theorem \ref{Kazarian} in Sect.\ref{Sec42} and Sect.\ref{Sec43} for completion.

Consider the following symmetric polynomials of degree $3(2g-2+l)$,
\begin{equation}\label{symmetricpoly}
\widehat{H}_{g,l}(z_1,\dots ,z_l)=\sum_{n_1,\dots,n_l}\left<\tau_{n_1}\dots\tau_{n_l}\Lambda_g^{\vee}(1)\right>_{g,l}\prod_{i=1}^l\phi_{n_i}(z_i),
\end{equation}
where $\Lambda_g^{\vee}(1)=1+\sum_{i=1}^g(-1)^i\lambda_i$. Then we have
\begin{Thm}[\cite{EMS},\cite{MZ}]\label{Thmrecursion}
The polynomials (\ref{symmetricpoly}) satisfy the following topological recursion relation
\begin{multline}\label{recursion}
\left(2g-2+l+\sum_{i=1}^l\frac{1}{z_i+1}D_i\right)\cdot \widehat{H}_{g,l}(z_L)\\
=\sum_{i<j}\frac{(z_i+1)^2z_jD_i\cdot\widehat{H}_{g,l-1}(z_{L/\{j\}})-(z_j+1)^2z_iD_j\cdot\widehat{H}_{g,l-1}(z_{L/\{i\}})}{z_i-z_j}\\
+\sum_{i=1}^l\left[D_{u_1}D_{u_2}\cdot\widehat{H}_{g-1,l+1}(u_1,u_2,z_{L/\{i\}}) \right]_{u_1=u_2=z_i}\\
+\frac{1}{2}\sum_{i=1}^l\sum_{\substack{g_1+g_2=g\\J\cup K=L/\{i\}}}^{\mbox{\scriptsize{stable}}}D_i\cdot\widehat{H}_{g_1,|J|+1}(z_i,z_J)D_i\cdot\widehat{H}_{g_2,|K|+1}(z_i,z_K).
\end{multline}
Here, $L=\{1,2,\dots,l\}$ is the index set, and for a subset $J\subset L$, we write $z_J=(z_j)_{j\in J}$. The last summation of the above equation is taken over all partitions $g=g_1+g_2$ of the genus $g$ and disjoint union decomposition $J\cup K=L/\{i\}$ satisfy the stability conditions $2g_1-1+|J|>0$ and $2g_2-1+|K|>0$.
\end{Thm}

\vspace{6pt}
\noindent
{\bf Remark:} In their original statements, the operator $D$ in the above theorem is in the form $t^2(t-1)\frac{\d}{\d t}$. In our case, we substitute $t$ by $z+1$. This does not effect the validity of the theorem.

\subsection{Equivalence relation between Kazarian's formula and the polynomial recursion }\label{Sec41}
In this subsection ,we give a brief explanation that Kazarian's formula and the polynomial recursion formula of linear Hodge integral in Theorem \ref{Thmrecursion} are actually equivalent. The equivalence relation between these two theorems is already guaranteed, because they are both equivalent to the cut-and-join equation (\ref{cutandjoin}) of the simple Hurwitz numbers (cf. \cite{MK},\cite{MZ}). But here, as mentioned before, we do not consider the theory of Hurwitz numbers. Instead, we show their equivalence relation in a more straightforward way by indicating the corresponding relation between the action of the operators in Eq.(\ref{Kazariancutandjoin}) and the terms in Eq.(\ref{recursion}).

\vspace{6pt}
\noindent
{\bf Remark:} For the function $F_H(u,q)$ defined in (\ref{FHuq}), after the rescaling $q_k\rightarrow u^kq_k$, we have
\begin{equation*}
\left.F_H(u,q)\right\vert_{q_k\rightarrow u^kq_k}=\sum_{g,d_i,n} \left< \tau_{d_1}\dots \tau_{d_n}\Lambda_g^{\vee}(1)\right>_{g,n}u^{3(2g-2+n)}\prod_{i=1}^n \widetilde{\phi_{d_i}}(u=1,q).
\end{equation*}
The contribution of the terms with the fixed genus $g$ and $n=l$ in the above equation is:
\begin{equation}\label{symmetricpolybefore}
\sum_{d_1,\dots,d_l} \left< \tau_{d_1}\dots \tau_{d_l}\Lambda_g^{\vee}(1)\right>_{g,l}u^{3(2g-2+l)}\prod_{i=1}^l \widetilde{\phi_{d_i}}(u=1,q).
\end{equation}
If we apply the operator
$\frac{1}{l!}\prod_{j=1}^l\left(\sum_{k=1}^{\infty}z_j^k\frac{\partial}{\partial q_k}\right)$
on Eq.(\ref{symmetricpolybefore}), and set $u=1$, then it gives us the symmetric polynomial $\widehat{H}_{g,l}(z_1,\dots ,z_l)$ in (\ref{symmetricpoly})， (such operators appeared in \cite{MZo}). Hence, in order to find equivalence relation between the Virasoro constraints and the recursion formula (\ref{recursion}) for linear Hodge integrals, the function $F_H(u,q)$ is a much better choice than the linear Hodge partition function $F_H(u,t)$.
\begin{flushright}
$\Box$
\end{flushright}

Let us transform (\ref{Kazariancutandjoin}) into
\begin{multline}\label{re-arrange}
\left(\frac{1}{3}u\frac{\partial}{\partial u}+\frac{1}{3}L_0+u^{-1}L_{-1}+L_0\right)\cdot \exp(F_H(u,q))\\
=\left(\vphantom{\frac{1}{2}}u^3M_0+4u^{2}M_{-1}+6uM_{-2}+4M_{-3}+u^{-1}M_{-4} \right.\\
\left.+\frac{1}{4}uq_2+\frac{1}{3}q_3+\frac{1}{8}u^{-1}q_4\right)\cdot \exp(F_H(u,q)).
\end{multline}
We will be looking at the following four operators, namely,
\begin{align*}
A_1 &= \frac{1}{3}u\frac{\partial }{\partial u}+\frac{1}{3}L_0, \quad\quad
A_2 = u^{-1}L_{-1}+L_0,\\
A_3 &= u^3M_0+4u^{2}M_{-1}+6uM_{-2}+4M_{-3}+u^{-1}M_{-4}-\frac{2}{3}q_1^3-\frac{1}{2}u^{-1}q_1^2q_2,\\
A_4 &= \frac{2}{3}q_1^3+\frac{1}{2}u^{-1}q_1^2q_2+\frac{1}{4}uq_2+\frac{1}{3}q_3+\frac{1}{8}u^{-1}q_4.
\end{align*}
Note that the first two terms of $A_4$ are the polynomial part of $4M_{-3}+u^{-1}M_{-4}$.

{\bf Operator $A_1$:}
First, we can transform Eq.(\ref{condition}) into $2j+\sum_{i=1}^n (2d_i+1)=3(2g-2+n).$
Suppose each term in $F_H(u,q)$ is in the form $u^xq_{b_1}\dots q_{b_n}$ with its corresponding coefficient. Then $2j+\sum_{i=1}^n (2d_i+1)=x+\sum_{i=1}^n b_i,$
which gives us $\frac{1}{3}(x+\sum_{i=1}^n b_i)=2g-2+n.$ Hence
\begin{equation}\label{term:2g-2+n}
\frac{1}{3}(u\frac{\partial }{\partial u}+L_0)\cdot F_H(u,q)= \sum_{g,n}(2g-2+n)\sum \left< \tau_{d_1}\dots \tau_{d_n}\lambda_j\right>_{g,n}u^{2j}\prod_{i=1}^n \widetilde{\phi_{d_i}}(u,q).
\end{equation}

{\bf Operator $A_2$:} We can see the corresponding relation by comparing the following two terms:
$$(u^{-1}L_{-1}+L_0)\cdot q_k= u^{-1}kq_{k+1}+kq_k,$$
$$\frac{1}{1+z}D\cdot z^k=\left(z^2\frac{\d}{\d z}+z\frac{\d}{\d z}\right)\cdot z^k=kz^{k+1}+kz^k.$$

{\bf Operator $A_3$:} Let $A_3=\mathcal{C}+\mathcal{J}$, where $\mathcal{C}$ is the ``cut'' part consisting of only first order differential operators, and $\mathcal{J}$ is the ``join'' part consisting of second order differential operators. Keep in mind that on the right hand side of Eq.(\ref{recursion}), there are three parts. The action of operator $\mathcal{C}$ on a single variable $q_m$ can be described as
\begin{multline}\label{cut}
\mathcal{C}\cdot q_m=m\bigg(\sum_{a+b=m}q_aq_b +4\sum_{a+b=m+1}q_aq_b +6\sum_{a+b=m+2}q_aq_b \\
 +4\sum_{a+b=m+3}q_aq_b +\sum_{a+b=m+4}q_aq_b  \bigg),
\end{multline}
which corresponds to the following terms in the first part of the right hand side of Eq.(\ref{recursion}):
\begin{multline*}
\frac{(z_i+1)^2z_jD_i\cdot z_i^m-(z_j+1)^2z_iD_j\cdot z_j^m}{z_i-z_j}\\
=m\bigg(\sum_{a+b=m}z_i^az_j^b +4\sum_{a+b=m+1}z_i^az_j^b +6\sum_{a+b=m+2}z_i^az_j^b+4\sum_{a+b=m+3}z_i^az_j^b +\sum_{a+b=m+4}z_i^az_j^b \bigg)
\end{multline*}
And, for $a,b\geq 1$, the action of operator $\mathcal{J}$ on $q_aq_b$ is in the form
\begin{equation}\label{join}
\mathcal{J}\cdot q_aq_b=ab\left(q_{a+b}+4q_{q+b+1}+6q_{q+b+2}+4q_{q+b+3}+q_{q+b+4}\right),
\end{equation}
which corresponds to the following terms in the second part:
\begin{equation*}
\left[D_{u_1}D_{u_2}\cdot u_1^au_2^b \right]_{u_1=u_2=z_i}
=ab\left(z_i^{a+b}+4z_i^{q+b+1}+6z_i^{q+b+2}+4z_i^{q+b+3}+z_i^{q+b+4}\right).
\end{equation*}
Note that $\mathcal{J}$ is a second order differential operator. We must take into account the following terms:
\begin{multline*}
\left(q_{a+b}+4q_{q+b+1}+6q_{q+b+2}+4q_{q+b+3}+q_{q+b+4}\right)\left(a\frac{\partial}{\partial q_a}q_a\right)\left(b\frac{\partial}{\partial q_a}q_b\right),
\end{multline*}
which corresponds to the following terms in the third part:
\begin{align*}
&\left(\vphantom{z_i^b}D_i\cdot z_i^a\right)\left(D_i\cdot z_i^b\right)\\
=&\left(z_i^{a+b}+4z_i^{q+b+1}+6z_i^{q+b+2}+4z_i^{q+b+3}+z_i^{q+b+4}\right)\left(z_i\frac{\partial}{\partial z_i}z_i^a\right)\left(z_i\frac{\partial}{\partial z_i}z_i^b\right).
\end{align*}

{\bf Operator $A_4$:}  Since all the contributions $2g-2+n=1$ in the recursion formula are from $(g,n)=(0,3)$ or $(1,1)$, the initial equation becomes
\begin{equation*} \left(A_1+A_2\right)\cdot \left( \left<\tau_0^3\right>_{0,3}q_1^3+\left<\tau_1\right>_{1,1}\widetilde{\phi_1}(u,q)-\left<\tau_0\lambda_1\right>_{1,1}u^2q_1 \right)+A_4=0.
\end{equation*}
Solving the above equation we get
$$\left<\tau_0^3\right>_{0,3}=\frac{1}{6}, \quad \left<\tau_1\right>_{1,1}=\left<\tau_0\lambda_1\right>_{1,1}=\frac{1}{24}$$
as expected.

Until now, we have established the corresponding relation between the action of the operators in Eq.(\ref{Kazariancutandjoin}) and the terms in Eq.(\ref{recursion}).

\subsection{From the Virasoro constraints and Dilaton equation to Kazarian's formula}\label{Sec42}
In this subsection, we prove Kazarian's formula using the Virasoro constraints and Dilaton equation for Hodge tau-functions in Theorem \ref{VforHodge}. A straightforward computation shows that
\begin{equation}\label{def:M4}
M_{-4}=\sum_{m=-1}^{\infty}q_{2m+4}L_{2m}+\sum_{k=1}^{\infty}L_{-2k-4}^{odd}\left(2k\frac{\partial}{\partial q_{2k}}\right),
\end{equation}
where
$$
L_{-2k}^{odd} := \sum_{\substack{j>0\\j \mbox{ \scriptsize{is odd}}}} q_{2k+j} j \frac{\partial}{\partial q_{j}} + \frac{1}{2}\sum_{\substack{i+j=2k\\j \mbox{ \scriptsize{is odd}}}} ij \frac{\partial^2}{\partial q_{i}\partial q_{j}} .
$$
Since
$$\left(L_{2m}-(2m+3)\frac{\partial}{\partial q_{2m+3}}+\frac{1}{8}\delta_{m,0}\right)\cdot \exp(F_K(q))=2k\frac{\partial}{\partial q_{2k}}\exp(F_K(q))=0,$$
for $m\geq -1, k\geq 1$, we can easily deduce that
\begin{equation}\label{M4}
\left(M_{-4}-L_{-1}+\frac{1}{8}q_4\right)\cdot \exp(F_K(q))=0.
\end{equation}
Motivated by above argument, we are able to give the parallel results for the Hodge tau function. Let
\begin{equation*}
\widetilde{P} := \sum_{i=3}^{\infty} \left[z^i\right]\left( -\frac{z^3}{(1+z)^3}+\frac{z^2}{(1+z)^2}\eta \right) u^{i-3}\frac{\partial}{\partial q_i}.
\end{equation*}
First we have
\begin{Prop} \label{F2a}
\begin{multline} \label{conj}
e^{\widetilde{P}}e^U \left(M_{-4}-L_{-1}+\frac{1}{8}q_4\right)  e^{-U}
e^{-\widetilde{P}} \\
=\sum_{m=-1}^{\infty} \left(e^U e^P q_{2m+4} e^{-P} e^{-U} \right) \cdot V_{2m}^{(H)}+\sum_{k=1}^{\infty}  \left(e^U e^P L_{-2k-4}^{odd}  e^{-P} e^{-U}   \right) \cdot E^{(k)} ，
\end{multline}
where $E^{(k)}$ is defined in Eq. \eqref{epsilon}
\end{Prop}
\begin{proof}
We claim that $$\widetilde{P}=e^U \left(\sum_{i=2}^{\infty}(-1)^ib_iu^{i-1}\frac{\partial}{\partial q_{i+2}}\right) e^{-U},$$
then the formula follow from
$$
e^{\widetilde{P}}e^U V^{(K)}_{2m}e^{-U}
e^{-\widetilde{P}}= V^{(H)}_{2m} \quad \text{and} \quad
e^{\widetilde{P}}e^U \left(2k \frac{\partial}{\partial_{q_{2k}}}\right)e^{-U}
e^{-\widetilde{P}}= E^{(k)} .
$$
The claim is proved by
\begin{align*}
&\exp\left(-\sum_{m=1}^{\infty}a_m(z^{1+m}\frac{\d}{\d z}+mz^m)\right)\cdot \sum_{i=2}^{\infty}(-1)^ib_iz^{i+2}
=-\frac{z^3}{(1+z)^3}+\frac{z^2}{(1+z)^2}\eta.\\
\end{align*}
\end{proof}
 
The following formula is a direct consequence of Virasoro constraints \eqref{maineq}, Proposition \ref{q2m} and the above proposition:
\begin{equation}\label{F1}
e^{\widetilde{P}}e^U\left(M_{-4}-L_{-1}+\frac{1}{8}q_4\right) e^{-U}e^{-\widetilde{P}} \cdot \exp(F_H(u,q))=0.
\end{equation}
Next, we aim to compute the conjugation on the left hand side of the equation \eqref{conj} using the following formulas
\begin{align*}
[q_n, L_k]&= -n\alpha_{k-n},\quad [\frac{\partial}{\partial q_n}, L_k]=\alpha_{n+k}, \quad [\frac{\partial}{\partial q_n}, M_k]=L_{n+k},\\
[L_m,L_n]& =(m-n)L_{m+n}+\frac{1}{12}(m^3-m)\delta_{m+n,0},\\
[L_n,M_k]& =(2n-k)M_{n+k} +\frac{n^3-n}{12}\alpha_{n+k},
\end{align*}
where
\[
\alpha_n=
\begin{cases}
q_{-n} & n<0 \\
0 &  n=0 \\
n\frac{\partial}{\partial q_n} &  n>0.
\end{cases}
\]

We show that
\begin{Prop}\label{F2}
\begin{multline*}
e^{\widetilde{P}}e^U\left(M_{-4}-L_{-1}+\frac{1}{8}q_4\right) e^{-U}e^{-\widetilde{P}}\\
=u^4\widetilde{M}+\frac{1}{3}uL_0+\sum_{i=2}^{\infty}[z^i]\frac{z^2}{(1+z)^3}u^{i-2}\frac{\partial}{\partial q_i}+\frac{1}{24}u
-\frac{1}{2}E^{(1)},
\end{multline*}
where $E^{(1)}$ is defined in Eq.\eqref{epsilon}.
\end{Prop}
\begin{proof}
First we consider $e^UM_{-4}e^{-U}$. By the formula of $[L_n,M_k]$, we know that there are two parts in $e^UM_{-4}e^{-U}$. For the ``cut-and-join'' part of $e^UM_{-4}e^{-U}$, using the method in the proof of Lemma \ref{pmH+LmH}, we can see that it is sufficient to consider the following power series:
\begin{align}\label{14641}
\exp(-\sum a_kz^{1+k}\frac{\d}{\d z}+2\sum ka_kz^{k})\cdot z^{-4}\nonumber
=1+\frac{4}{z}+\frac{6}{z^2}+\frac{4}{z^3}+\frac{1}{z^4}.
\end{align}
This implies that the ``cut-and-join'' part of $e^UM_{-4}e^{-U}$ is exactly
\begin{equation}\label{uA3}
uA_3=u^4M_0+4u^3M_{-1}+6u^2M_{-2}+4uM_{-3}+M_{-4}.
\end{equation}
The remaining part of $e^UM_{-4}e^{-U}$ is (see Appendix A.3)
\begin{equation}\label{Q3}
\frac{1}{4}u^2q_2+\frac{1}{3}uq_3+\frac{1}{8}q_4-e^U(\frac{1}{8}q_4)e^{-U}.
\end{equation}
The operator $e^U(-L_{-1})e^{-U}$ follows from Lemma \ref{pmH+LmH}, that is,
\begin{equation}\label{L1}
e^U(-L_{-1})e^{-U}=-\sum_{i=-1}^{\infty}[z^i]\left( \frac{(1+z)^2}{z^2}\eta \right) u^{i+1}L_i.
\end{equation}
The operator $e^{\widetilde{P}}(uA_3)e^{-\widetilde{P}}$ consists of three parts. The ``cut-and-join'' part is still $uA_3$. The Virasoro operator part is $[\widetilde{P}, uA_3]$, that is,
\begin{align}\label{PtildeA3}
[\widetilde{P}, uA_3]=&\sum_{i=0}^{\infty}[z^i]\left(-\frac{z^3}{(1+z)^3}+\frac{z^2}{(1+z)^2}\eta\right)\left(1+\frac{1}{z}\right)^4 u^{1-i}L_i\nonumber\\
=&-uL_0-L_{-1}+e^UL_{-1}e^{-U}.
\end{align}
The third part is the first order differential operator $\frac{1}{2}[\widetilde{P}, [\widetilde{P}, uA_3]]$, which is in the form
\begin{align}\label{Q1}
&\sum_{i=2}^{\infty}[z^i]\frac{1}{2}\left(-\frac{z^3}{(1+z)^3}+\frac{z^2}{(1+z)^2}\eta\right)\left(-1-\frac{1}{z}+ \frac{(1+z)^2}{z^2}\eta\right)u^{i-2}\alpha_i\nonumber\\
=&\sum_{i=2}^{\infty}[z^i]\left(-\frac{z}{1+z}\eta+\frac{1}{2}\eta^2\right)u^{i-2}\alpha_i+\sum_{i=2}^{\infty}[z^i]\frac{z^2}{(1+z)^3}u^{i-2}\frac{\partial}{\partial q_i}.
\end{align}
And the operator $e^{\widetilde{P}}e^U(-L_{-1})e^{-U}e^{-\widetilde{P}}$ contains an extra first order differential operator, which is in the form
\begin{align}\label{Q2}
&\sum_{i=2}^{\infty}[z^i]\left(-\frac{z^3}{(1+z)^3}+\frac{z^2}{(1+z)^2}\eta\right)\left(-\frac{(1+z)^2}{z^2}\eta\right)u^{i-2}\alpha_i\nonumber\\
=&\sum_{i=2}^{\infty}[z^i]\left(\frac{z}{1+z}\eta-\eta^2 \right)u^{i-2}\alpha_i.
\end{align}
The operator $e^{\widetilde{P}}e^U(\frac{1}{8}q_4)e^{-U}e^{-\widetilde{P}}$ is $e^U(\frac{1}{8}q_4)e^{-U}$ plus the constant $\frac{1}{24}u$. Finally, we sum up equations from (\ref{uA3}) to (\ref{Q2}) and the constant  $\frac{1}{24}u$. This gives us
\begin{multline*}
u^4M_0+4u^3M_{-1}+6u^2M_{-2}+4u M_{-3}+M_{-4}-uL_0-L_{-1}\\
+\frac{1}{4}u^2q_2+\frac{1}{3}uq_3+\frac{1}{8}q_4 +\sum_{i=2}^{\infty}[z^i]\frac{z^2}{(1+z)^3}u^{i-2}\frac{\partial}{\partial q_i}+\frac{1}{24}u
-\frac{1}{2}\sum_{n=2}^{\infty}[z^n]\eta^2 n\frac{\partial}{\partial q_n},
\end{multline*}
and completes the proof of the lemma.
\end{proof}

Now we finalize the proof of Theorem \ref{Kazarian}. 
Note that
\begin{align*}
&\sum_{i=2}^{\infty}[z^i]\frac{z^2}{(1+z)^3}u^{i-2}\frac{\partial}{\partial q_i}+\sum_{k=0}^{\infty}(-1)^k\binom{k+2}{k}u^{k+1}\frac{\partial}{\partial q_{k+3}}= E^{(1)}.
\end{align*}
Then, by Eq.(\ref{F1}), Proposition \ref{F2} and Proposition \ref{q2m}, we have
\begin{align*}
&\left(u^4\widetilde{M}-\frac{1}{3}u^2\frac{\partial}{\partial u}\right)\cdot \exp(F_H(u,q))=0.
\end{align*}
which leads us to Kazarian's formula.

\subsection{From Kazarian's formula to the Virasoro constraints and Dilaton equation}\label{Sec43}
In this subsection, we briefly explain how to derive the Virasoro constraints (\ref{maineq}) and Eq.\eqref{dilaton3} using only Kazarian's formula. First, we prove Proposition \ref{epsiloncommutator}. Our computation using the power series shows that
\begin{align*}
&[u^{-3}A_3,u^{2k}E^{(k)}]=-2\sum_{i=2k-4}^{\infty} \nu_i^{(k-2)}u^{i-3}L_i,\\
&[-u^{-3}L_0-u^{-4}L_{-1},u^{2k}E^{(k)}]=2\sum_{j=2m-1}^{\infty}\gamma_j^{(m-2)}ju^{j-3}\frac{\partial}{\partial q_j}\\
&[\frac{1}{4}u^{-2}q_2+\frac{1}{3}u^{-3}q_3+\frac{1}{8}u^{-4}q_4,u^{2k}E^{(k)}]=-\frac{1}{4}\delta_{k,2}\\
&[-\frac{1}{3}u^{-2}\frac{\partial}{\partial u}-\frac{1}{3}u^{-3}L_0,u^{2k}E^{(k)}]=0.
\end{align*}
The summation of the above for equations gives us the proposition. Hence, by Theorem \ref{Kazarian} and Proposition \ref{q2m}, we can conclude that
\begin{equation*}
 V_{m-2}^{(H)} \cdot \exp(F_H(u,q)) = -\frac{1}{2}[\widetilde{M}-\frac{1}{3}u^{-2}\frac{\partial}{\partial u},u^{2m}E^{(m)} ] \cdot \exp(F_H(u,q))=0, \quad (m\geq 1).
\end{equation*}
This proves the Virasoro constraints. Eq.\eqref{dilaton3} can be deduced from Eq.\eqref{F1} (which are deduced from the Virasoro constraints that we already proved) and Proposition \ref{F2} (which are independent results) directly.

If we continue the path from $\exp(F_H(u,q))$ to $\exp(F_K(q))$ using Theorem \ref{KtoHodge}, we will land on the Virasoro constraints of Kontsevich-Witten tau-function with no surprises.

\section{More remarks}
Let $H$ be the generating function of simple Hurwitz numbers, (we skip the definition of simple Hurwitz numbers, which can be found in many papers such as \cite{MK}). The function $\exp(H)$ has a cut-and-join representation \cite{GJ}.
\begin{equation}\label{cutandjoin}
e^H=e^{\beta M_0}\cdot e^{q_1},
\end{equation}
where
$$M_0=\frac{1}{2}\sum_{i,j>0}\left((i+j)q_iq_j\frac{\partial}{\partial q_{i+j}}+ijq_{i+j}\frac{\partial^2}{\partial q_i\partial q_j}\right).$$
Since the cut-and-join operator $M_0$ belongs to $\widehat{\mathfrak{gl}(\infty)}$, the function $\exp(H)$ is a tau-function for the KP hierarchy. We call it the Hurwitz tau-function.

In \cite{MK}, Kazarian obtained the function $F_H(u,q)$ we have introduced before from $H$ using the ELSV formula (see \cite{ELSV}). In fact, setting $u=\beta^{\frac{1}{3}}$, we have (Theorem 2.3 in \cite{MK})
\begin{multline*}
\exp(-\sum_{m=1}^{\infty} a_{-m}\beta^m\sum_{i=1}^{\infty} iq_{i+m}\frac{\partial}{\partial q_i})
\cdot (H-H_{0,1}-H_{0,2})=\exp(\frac{4}{3}\log{\beta} L_0)\cdot F_H(\beta^{\frac{1}{3}},q),
\end{multline*}
where
$$H_{0,1}=\sum_{b=0}^{\infty}\frac{b^{b-2}}{b!}\beta^{b-1}q_b, \quad H_{0,2}=\frac{1}{2}\sum_{b_1,b_2=1}^{\infty}\frac{b_1^{b_1}b_2^{b_2}}{(b_1+b_2)b_1!b_2!}\beta^{b_1+b_2}q_{b_1}q_{b_2}.$$
And the numbers $\{a_{-m}\}, m>0,$ are determined by
\begin{equation*}
\exp(-\sum_{m=1}^{\infty} a_{-m}z^{1+m}\frac{\partial}{\partial z})\cdot z=\frac{z}{1+z}e^{-\frac{z}{1+z}}.
\end{equation*}
Lemma 4.5 in \cite{MK} actually implies that
\begin{equation*}
\exp(-\sum_{m=1}^{\infty} a_{-m}\beta^mL_{-m})=\exp(-\sum_{m=1}^{\infty} a_{-m}\beta^mX_{-m})\exp(-H_{0,2}).
\end{equation*}
This can be shown using the method introduced in \cite{LW} and Appendix A.2. Hence, if we let
\begin{equation*}
F_H(\beta)=\exp(\frac{4}{3}\log\beta L_0)\cdot F_H(\beta^{\frac{1}{3}},q),
\end{equation*}
Then we can deduce that
\begin{equation*}
\exp(H)=\exp(H_{0,1})\exp(\sum_{m>0} a_{-m}\beta^mL_{-m})\cdot \exp(F_H(\beta)).
\end{equation*}
The Virasoro constraints for the tau-function $\exp(F_H(\beta))$ are $\{V_m^{(H)}\}$ with 
$u=\beta^{1/3}, q_k\rightarrow \beta^{4k/3}q_k$. As mentioned in \cite{A}, one may wish to derive a set of Virasoro constraints for Hurwitz tau-function $\exp(H)$ by a simple conjugation of $V_m^{(H)}$ using the operator $\exp(\sum_{m>0} a_{-m}\beta^mL_{-m})$. However, since $V_m^{(H)}$ is an infinite linear combination of Virasoro operator $L_i$ with $i\geq 2m$, it is very likely that such conjugation will cause a divergence problem. And it seems impossible to obtain a constraint involving only finite number of $L_i$ by a certain linear combination of $V_m^{(H)}$.

\section*{Appendix}
\subsection*{A.1 String and dilaton equations}
Here we give a brief argument of the string equation: for $2g-2+n>0$, $1\leq j\leq g$,
\begin{equation}\label{string2}
<\lambda_j \tau_0\prod_{i=1}^n\tau_{d_i}>=\sum_{k=1}^n<\lambda_j \tau_{d_k-1}\prod_{i\neq j}\tau_{d_i}>.
\end{equation}
and dilaton equation: for $2g-2+n>0$, $1\leq j\leq g$,
\begin{equation}\label{dilaton}
<\lambda_j \tau_1\prod_{i=1}^n\tau_{d_i}>=(2g-2+n)<\lambda_j \prod_{i=1}^n\tau_{d_i}>.
\end{equation}
For the case $m=-1$ in formula (\ref{maineq}), we have
\begin{multline}\label{string1}
V_{-1}^{(H)}\cdot\exp(F_H(u,q))\\
=\left(L_{-2}+2uL_{-1}+u^2L_0-\frac{u^2}{24}-\sum_{j=1}^{\infty}(-u)^{j-1}j\frac{\partial}{\partial q_j}\right)\cdot\exp(F_H(u,q))=0.
\end{multline}
Observe that, for the polynomial $\phi_k(z)$, we have
\[
\phi_k(-1) =
\begin{cases}
-1 & k=0 \\
0 & k\geq 1 .
\end{cases}
\]
Then,
\[
\sum_{j=1}^{\infty}(-u)^{j-1}j\frac{\partial}{\partial q_j} \widetilde{\phi_{k}}(u,q) =
\begin{cases}
1 & k=0 \\
0 & k\geq 1 .
\end{cases}
\]
Taking into account Eq.(\ref{def:phi}), we can obtain the following equation from (\ref{string1}): for $2g-2+n>0$,
\begin{align*}
&\sum_{\substack{1\leq j\leq g\\d_i\geq 0}}(-1)^j<\lambda_j \tau_0\prod_{i=1}^n\tau_{d_i}>u^{2j}\prod_{i=1}^n\widetilde{\phi_{d_i}}(u,q) \nonumber\\
=&\sum_{\substack{1\leq j\leq g\\d_i\geq 0}}\sum_{k=1}^n(-1)^j<\lambda_j \tau_{d_k-1}\prod_{i\neq j}\tau_{d_i}>u^{2j}\prod_{i=1}^n\widetilde{\phi_{d_i}}(u,q).
\end{align*}
Since $\{\widetilde{\phi_{k}}(u,q)\}$ is a linear independent set for $k\geq 0$, the above equation immediately implies the string Eq.(\ref{string2}).

Next we explain how Eq.\eqref{dilaton3} gives us the dilaton equation. Let
$$c_n=\sum_{k=0}^{\infty}(-1)^k\binom{k+2}{k}u^{k}\frac{\partial}{\partial q_{k+3}}\cdot \widetilde{\phi_n}(u,q).$$
The cases $n=0$ and $n=1$ are obvious. For $n\geq 2$, we notice that
$$c_n=[z^{0}] \left\{\frac{z^3}{(1+z)^3}\left(\phi_n(z^{-1})\right)\right\}.$$
Then, from Eq.(\ref{cn}), we have
\[
c_n=
\begin{cases}
1 & n=1 \\
0 & n\neq 1 .
\end{cases}
\]
Hence, using Eq.(\ref{term:2g-2+n}) and the above result, we can obtain the following equation from Eq.(\ref{dilaton3}): for $2g-2+n>0$,
\begin{align*}
&\sum_{\substack{1\leq j\leq g\\d_i\geq 0}}(-1)^j<\lambda_j\tau_1\prod_{i=1}^n\tau_{d_i}>u^{2j}\prod_{i=1}^n\widetilde{\phi_{d_i}}(u,q)\\
=&(2g-2+n)\sum_{\substack{1\leq j\leq g\\d_i\geq 0}}(-1)^j<\lambda_j\prod_{i=1}^n\tau_{d_i}>u^{2j}\prod_{i=1}^n\widetilde{\phi_{d_i}}(u,q).
\end{align*}
This proves the dilaton equation Eq.(\ref{dilaton}).

\subsection*{A.2 Some technical results}
We recall a method introduced in \cite{LW}. For $k=1,2,\dots$, let
$$\mathcal{D}_{-k}=z^{1+k}\frac{\d}{\d z}-kz^k, \quad A(u)=-\sum_{k=1}^{\infty}a_ku^k\mathcal{D}_{-k}.$$
Then
$$\exp(A(u))=\exp(-\sum a_ku^kz^{1+k}\frac{\partial}{\partial z})\exp(g(uz)).$$
It is easy to verify that, for $m,n>0$,
$[\mathcal{D}_{-m},\mathcal{D}_{-n}]=(n-m)\mathcal{D}_{-m-n}.$
Then there exists a unique sequence of numbers $\{d_n\}$, such that
$$\frac{\partial}{\partial u}e^{A(u)}=(\sum_{n=1}^{\infty}d_nu^{n-1}\mathcal{D}_{-n}) e^{A(u)}.$$
On the other hand, we have
$$\frac{\partial}{\partial u}e^{A(u)}= \left(\sum_{n=1}^{\infty}d_nu^{n-1}z^{1+k}\frac{\partial}{\partial z}+ \exp(-\sum_{k=1}^{\infty}a_ku^kz^{1+k}\frac{\partial}{\partial z})\cdot \frac{\partial}{\partial u} g(uz) \right) e^{A(u)}.$$
This gives us
$$\frac{\partial}{\partial u} g(uz)=\exp(\sum_{k=1}^{\infty}a_ku^kz^{1+k}\frac{\partial}{\partial z})\cdot \left( -\sum_{n=1}^{\infty} d_nu^{n-1}nz^n\right).$$
And
\begin{equation}\label{gz}
z\frac{\d}{\d z}g(z)=\exp(\sum a_kz^{1+k}\frac{\d}{\d z})\cdot\left(-\sum_{n=1}^{\infty}d_{n}nz^{n}\right).
\end{equation}
We refer the readers to \cite{LW} for more details about the above argument.

\begin{Lem}\label{Lemphinz}
$$\phi_{n}(z) =\sum_{i=0}^n(-1)^i(2n-2i-1)!!C_i\left(f^{2n-2i+1}\right)_{+}.$$
\end{Lem}
\begin{proof}
First, we know that $z=(f)_{+}$. By induction, assume that, for $n=k$, $$\phi_{k}(z)=\sum_{i=1}^k(-1)^i(2k-2i-1)!!C_i(f^{2k-2i+1})_{+}.$$
When $n=k+1$, we have, by Eq.(\ref{f}),
\begin{align*}
\phi_{k+1}(z) &= (2k+1)!!\left(f^{2k+3}\right)_{+}-\sum_{i=1}^{k+1} C_i\phi_{k+1-i}(z)\\
&=(2k+1)!!\left(f^{2k+3}\right)_{+}-\sum_{n=1}^{k+1}(2k-2n+1)!!\left(f^{2k-2n+3}\right)_{+}\sum_{j=0}^{n-1}(-1)^jC_jC_{n-j}\\
&=\sum_{n=0}^{k+1}(-1)^n(2k-2n+1)!!C_n\left(f^{2k-2n+3}\right)_{+}.
\end{align*}
This completes the proof.
\end{proof}

\subsection*{A.3 More proof of Proposition \ref{F2}}
We consider the correspondence $L_n\rightarrow l_n, M_k\rightarrow m_k, n\partial q_n\rightarrow z^n$, where
\begin{align*}
l_n &= -z^{1+n}\frac{\d}{\d z}-\frac{1}{2}nz^n-z^n\\
m_k & =\frac{1}{2}z^k (z\frac{\d}{\d z}+\frac{1}{2})^2+\frac{(k+1)}{2}z^k(z\frac{\d}{\d z}+\frac{1}{2})+ \frac{1}{12}(k+1)(k+2)z^k.
\end{align*}
Then
\begin{align*}
& [\frac{1}{n}z^n,l_k]=z^{n+k} \quad\mbox{and}\quad [\frac{1}{n}z^n,m_k]=l_{n+k},\\
& [l_m,l_n]=(m-n)l_{m+n}+\frac{m^3-m}{12}\delta_{m+n,0},\\
& [l_n,m_k]=(2n-k)m_{n+k}+\frac{n^3-n}{12}z^{n+k}.
\end{align*}
These formulas agree with the formula presented in Sect.\ref{Sec42}. Let 
$$\Phi_z^{+} =\exp(\sum_{m=1}^{\infty}a_mz^{1+m}\frac{\d}{\d z}).$$
Then, 
\begin{align*}
\exp(\sum_{m=1}^{\infty}a_m (-z^{1+m}\frac{\d}{\d z}-\frac{1}{2}mz^m-z^m)) &= e^{-\Phi_z^{+}}\frac{1}{(1+h)}=\frac{1}{1+z}e^{-\Phi_z^{+}}.
\end{align*}
Now, by the definition of $m_k$, we have
\begin{equation*}
m_{-4}=\frac{1}{2}z^{-4}(z\frac{\d}{\d z})^2-z^{-4}(z\frac{\d}{\d z})-\frac{1}{8}z^{-4}.
\end{equation*}
From Eq.(\ref{14641}), we can see that
\begin{multline*}
\frac{1}{1+z}e^{-\Phi_z^{+}}m_{-4}e^{\Phi_z^{+}}(1+z)=m_0+4m_{-1}+6m_{-2}+4m_{-3}+m_{-4}+Q(z),
\end{multline*}
where $Q(z)$ is a polynomial. To compute $Q(z)$, we first compute
\begin{align*}
& \frac{1}{1+z}\frac{1}{2}\frac{(1+z)^4}{z^4} \left(z\frac{\d}{\d z}\right)^2 (1+z)\\
=& \frac{1}{2}\frac{(1+z)^4}{z^4}\left(z\frac{\d}{\d z}\right)^2+ \frac{(1+z)^3}{z^3}\left(z\frac{\d}{\d z}\right)+\frac{(1+z)^3}{2z^3}
\end{align*}
Then we have
\begin{align*}
 & e^{-\Phi_z^{+}}m_{-4} e^{\Phi_z^{+}}
 =\frac{1}{2} \frac{(1+z)^4}{z^4}\left(z\frac{\d}{\d z}\right)^2-\frac{(1+z)^3}{z^4}\left(z\frac{\d}{\d z}\right)-\frac{1}{8}\eta^{-4}.
\end{align*}
Hence
\begin{align*}
& \frac{1}{1+z}e^{-\Phi_z^{+}}m_{-4} e^{\Phi_z^{+}}(1+z) \\
=&m_0+4m_{-1}+6m_{-2}+4m_{-3}+m_{-4}+\frac{1}{8}z^{-4}+\frac{1}{3}z^{-3}+\frac{1}{4}z^{-2}-\frac{1}{24}-\frac{1}{8}\eta^{-4}.
\end{align*}
This shows that
$$Q(z)=\frac{1}{8}z^{-4}+\frac{1}{3}z^{-3}+\frac{1}{4}z^{-2}-\frac{1}{24}-\frac{1}{8}\eta^{-4}.$$
And therefore, 
\begin{equation*}
e^U(\frac{1}{8}q_4)e^{-U}=\frac{1}{8}\sum_{i=-4}^{\infty} [z^{i}](\eta^{-4})u^{i+4}\alpha_{i}.
\end{equation*}

\vspace{6pt}
\noindent
{\bf Remark:} It is definitely possible to prove Proposition \ref{F2} directly using the operators $l_n$ and $m_k$. But we think the use of power series like (\ref{gz}) and (\ref{14641}) simplifies a lot of computations, while the direct computation still needs to take care of the nested commutators of differential operators involving $l_n$ and $m_k$.

\vspace{10pt} \noindent
\footnotesize{\sc shuai guo:
school of mathematical sciences, peking University, beijing, china.}\\
\footnotesize{E-mail address:  guoshuai@math.pku.edu.cn}\\
\\
\footnotesize{\sc gehao wang: 
school of mathematics (zhuhai), sun yat-sen university, zhuhai, china. }\\
\footnotesize{E-mail address:  gehao\_wang@hotmail.com}

\end{document}